\documentclass[aps,twocolumn,showpacs,superscriptaddress]{revtex4}
\usepackage{graphicx,color}
\usepackage{indentfirst}
\usepackage{amsfonts}
\usepackage{subfigure}
\expandafter\let\csname equation*\endcsname\relax
\expandafter\let\csname endequation*\endcsname\relax
\usepackage{amsthm}
\usepackage{amsmath}
\usepackage{amssymb}
\usepackage{centernot}
\usepackage{bbold}
\usepackage{mathtools}
\usepackage[T1]{fontenc}
\usepackage{dsfont}
\usepackage{color}
\usepackage{braket}

\newcommand{\borb}[2]{\left | #1 \rangle\langle #2 \right |}

\newtheorem{theorem}{Theorem}
\newtheorem{conjecture}[theorem]{Conjecture}
\newtheorem{lemma}[theorem]{Lemma}
\newtheorem{proposition}[theorem]{Proposition}
\newtheorem{corollary}[theorem]{Corollary}

\begin{document}

\title{{Universality of finite time disentanglement}}

\author{Raphael C. Drumond}
\email{raphael@mat.ufmg.br}
\affiliation{Departamento de Matem\'{a}tica, Instituto
de Ci\^{e}ncias Exatas, Universidade Federal de Minas Gerais, CP
702, CEP 30123-970, Belo Horizonte, Minas Gerais, Brazil.}
\author{Cristhiano  Duarte}
\email{cristhiano@mat.ufmg.br}
\affiliation{Departamento de Matem\'{a}tica, Instituto de
Ci\^{e}ncias Exatas, Universidade Federal de Juiz de Fora, CEP 36036-330, 
Juiz de Fora, Minas
Gerais, Brazil.}
\affiliation{Departamento de Matem\'{a}tica, Instituto
de Ci\^{e}ncias Exatas, Universidade Federal de Minas Gerais, CP
702, CEP 30123-970, Belo Horizonte, Minas Gerais, Brazil.}
\author{Marcelo  Terra Cunha}
\email{tcunha@ime.unicamp.br}
\affiliation{Departamento de Matem\'{a}tica, Instituto
de Ci\^{e}ncias Exatas, Universidade Federal de Minas Gerais, CP
702, CEP 30123-970, Belo Horizonte, Minas Gerais, Brazil.}
\affiliation{Departamento de Matem\'{a}tica Aplicada, Instituto
de Matem\'atica, Estat\'istica e Computa\c c\~ao Cient\'ifica, Universidade 
Estadual de Campinas, Cidade Universit\'aria Zeferino Vaz,
CEP 13083-970, Campinas, S\~ao Paulo, Brazil.}
\author{M. C. Nemes}
\thanks{Deceased}
\affiliation{Departamento de F\'isica, Universidade Federal de Minas Gerais,
CP 702, CEP 30123-970, Belo Horizonte, MG - Brazil
}
\date{\today}

\begin{abstract}
In this paper we investigate how common is the phenomenon of Finite Time Disentanglement (FTD) with respect
to the set of quantum dynamics of bipartite quantum states with finite dimensional Hilbert spaces. 
Considering a quantum dynamics from a general sense, as just a continuous family of  Completely Positive Trace Preserving maps (parametrized by  the time variable) acting on the space of the bipartite systems, we conjecture that FTD happens for all dynamics but those when all maps of the 
family are induced by local unitary operations. 
We prove this conjecture valid for two important cases: i)  when all maps are induced by unitaries;
ii) for pairs of qubits, when all maps are unital. 
Moreover, we prove some general results about unitaries/CPTP maps preserving 
product/pure states
\end{abstract}


\maketitle

\section{Introduction}

Following the definition of entanglement as a resource for non-local
tasks, as a consequence being quantified~\cite{entanquant}, the time evolution of this
quantity was the subject of intense interest. Typically
a composite system will lose its entanglement whenever its parts interact 
with an
environment. It is of great interest then for practical implementations
of quantum
information protocols, that require entanglement, to understand how the amount of entanglement
behaves in time~\cite{review}. 

One characteristic of entanglement dynamics that drew a lot of attention was the possibility of an
initially entangled state to lose all its entanglement in a finite time, 
instead of asymptotically.
The phenomenon was initially called ``entanglement sudden 
death''~\cite{primeiramorte}, or Finite Time
Disentanglement (FTD). 
The simplest explanation for this fact is essentially 
topological:  for finite dimensional Hilbert spaces, the set
$\mathcal{S}$
of separable states, where 
entanglement is null, has non-empty interior, \emph{i.e.}, there are
``balls''
entirely consisted of separable states. 
Therefore, whenever an initially entangled state
approaches a separable state in the interior of $\mathcal{S}$, and given 
that the dynamics of the
state is continuous, it must spend at least a finite amount of time inside the set, so entanglement
will be null during this time interval~\cite{terra}. 

In references~\cite{terra2, terra3}, the authors explored how typical the 
phenomenon is {(for
several paradigmatic dynamics of two qubits and two harmonic 
oscillators)} when one varies the
initial
states for a fixed dynamics. 
Here we shall explore how 
typical it is with respect to the dynamics
themselves. 
More explicitly, given a dynamics for a composite system, should one expect to find some
initially entangled state exhibiting FTD? 
Here we argue that the answer is generally positive. 

The paper is organized as follows. In Section~\ref{ftd} we discuss about 
the generic existence of FTD and illustrate this discussion with a 
well-known example of a family of maps. In Section~\ref{tech} we go to the 
technical Lemmas and Theorems already used on Section~\ref{ftd}. 
We close this work with Section~\ref{Disc}, discussing further questions 
and open problems.     

\section{Finite time disentanglement}\label{ftd}

In a very broad sense, we can think a \emph{(continuous time) quantum dynamical system} as 
given by a family of completely positive trace preserving (CPTP) maps $\Lambda_{t}$, parametrized 
by the real time variable
$t$ for, say, $t \geq 0$. 
If a quantum system is in some state given by a 
density operator
$\rho_{0}$ at
$t=0$, for any {$t\geq0$} we have the system at the quantum 
state 
$\rho(t)=\Lambda_{t}(\rho_{0})$.
Of course, one must have $\Lambda_{0}=I$, where $I$ is the identity 
map. 
Although
in
some cases a discontinuous family of maps can be a good approximation to 
describe a process (for
example, when a very fast operation is performed on a system, or when the system will not be accessed during some time interval), strictly 
speaking the family of maps should be at least continuous. 

Generally speaking, fixed some
dynamics $\Lambda_{t}$, we say that it shows \emph{finite time 
disentanglement} (FTD) if there
exists an entangled state $\rho_{\text{ent}}$ and a time interval $(a,b)$, 
with $0<a<b\leq \infty$
such that $\Lambda_{t}(\rho_{\text{ent}})$ is a separable state  for all 
$t\in (a,b)$. In
Refs.~\cite{terra, terra2}, the authors point out that the occurrence of 
such effect is a
natural consequence of the set of separable states $\mathcal{S}$ having a 
non-empty interior.
Indeed, if an initially entangled state is mapped at some time $\bar{t}$ to 
a state in the interior
of $\mathcal{S}$, given the dynamics continuity, it must spend some finite 
time inside $S$ to reach
that state. During that time interval entanglement is null, although 
initially the system had
some entanglement. We shall formally state this fact for future reference:
\begin{proposition}
If a bipartite quantum dynamical system is such that, for some $\bar{t}>0$, 
there exists an 
initially entangled state $\rho_{ent}$ where its evolved state at time 
$\bar{t}$ is 
in the interior of the separable states, there is FTD.\label{fact}
\end{proposition}

This proposition is one of the main reasons of why we believe the following 
general conjecture is
valid:

\begin{conjecture} Given a bipartite quantum dynamical system with 
finite dimensional
Hilbert space $\mathcal{H}_{A}\otimes \mathcal{H}_{B}$ and a continuous 
family of CPTP maps
$\Lambda_{t}$, there is no finite time
disentanglement if, and only if, for all $t>0$ there exists unitary 
operations $U_{A,t}$ and
$U_{B,t}$ acting on $\mathcal{H}_{A}$ and $\mathcal{H}_{B}$, respectively,
such that
$\Lambda_{t}(\cdot)=(U_{A,t}\otimes U_{B,t})(\cdot)(U_{A,t}\otimes 
U_{B,t})^{*}$. \label{conjecture}
\end{conjecture} 

In physical terms, this says that FTD do not takes place only in the 
extremely special situation
where the pair of systems is closed (or at most interacting with a classical 
external field)
\emph{and} non-interacting. That is, whatever interaction they
may have, with each other \emph{or} with a third quantum system (such as a 
reservoir), FTD takes
place for some entangled state. From now on, we denote 
{the} 
family of dynamics {contained in 
Conjecture~\ref{conjecture}} by $\mathcal{F}_{\mathcal{H}_A,\mathcal{H}_B}$, 
that is: 
\begin{align}
\mathcal{F}_{\mathcal{H}_A,\mathcal{H}_B}&=\{ 
\{\Lambda_t(.)\}_{t \geq 0}; 
\,\, \{\Lambda_t(.)\}_{t \geq 0} \,\, \text{is continuous and}  \nonumber \\
&\Lambda_{t}(\cdot)=(U_{A,t}\otimes U_{B,t})(\cdot)(U_{A,t}\otimes 
U_{B,t})^{*}\}.
\end{align}

{Once again,} the intuition behind Conjecture 
\ref{conjecture} is geometrical. 
Figure~\ref{fig1} shows a pictorial 
representation of the set of quantum states when the Hilbert space is finite 
dimensional, with the 
distinguishing property of the set of separable states having non-empty 
interior. In
Figure~\ref{fig2} the arrows 
indicates the mapping of initial states to their corresponding evolved ones, 
on an instant of time 
$\bar{t}>0$. Note that all CPTP maps must have at least one fixed point, 
and 
all other states can 
not increase their distance to that fixed one, therefore for each instant 
of time $t \geq 0$ we can identify a 
``direction'' for the 
flow of states.  It is expected that if the flow is directed towards a 
separable state, some 
entangled states will be mapped inside the separable set (\ref{fig2}a). But 
even in the case where 
the flow is directed towards an entangled one, if the displacement is small 
enough, some entangled 
state located ``behind'' the set of separable  states will be mapped inside 
it (\ref{fig2}b). 
Below we prove this statement under some special conditions.

\begin{figure}
\includegraphics[width=5.0cm]{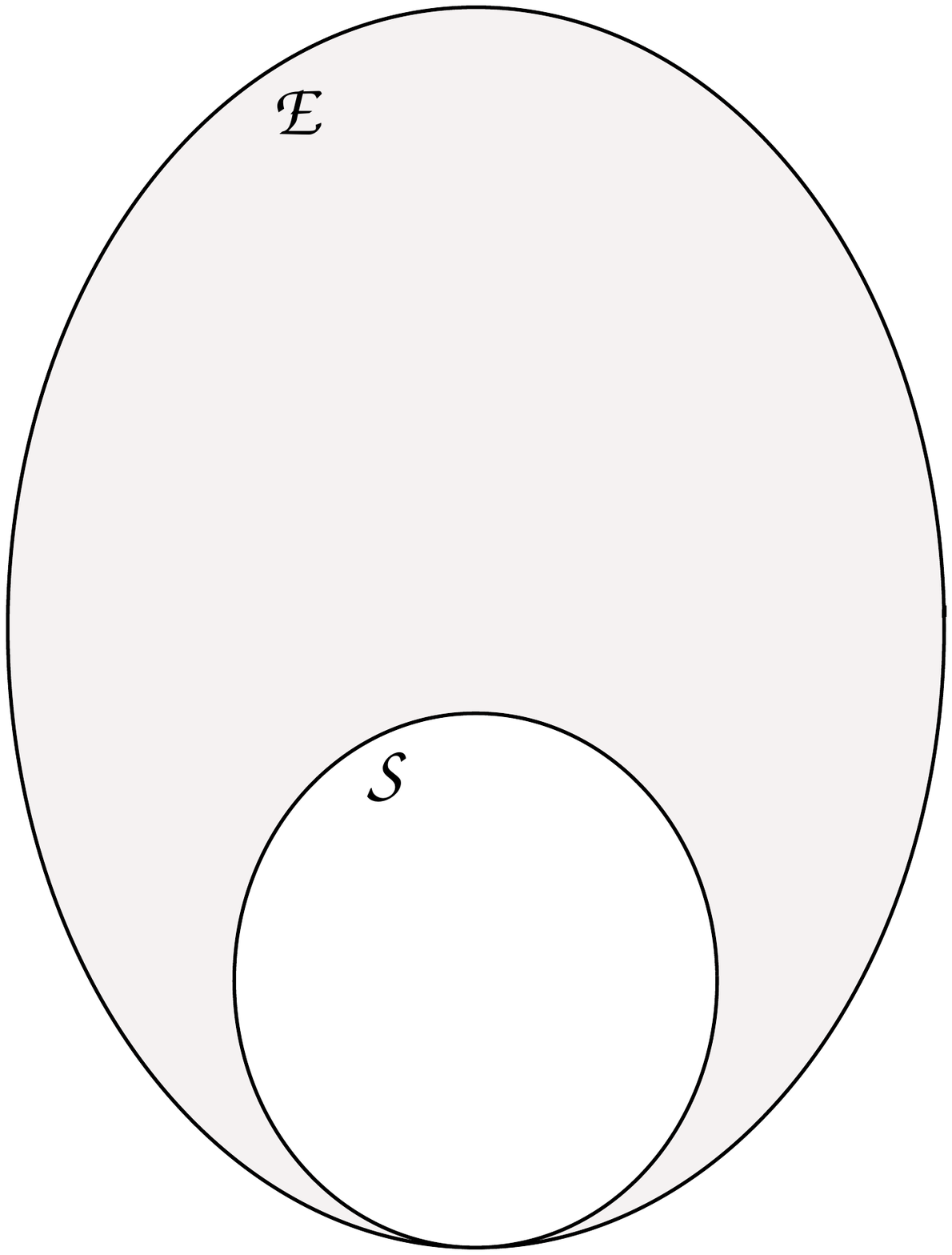}
\caption{Pictorial representation of set of quantum states when 
dim($\mathcal{H})<\infty$.}
\label{fig1}
\end{figure}

\begin{figure}[h]

\subfigure[]{\includegraphics[width=3.5cm]{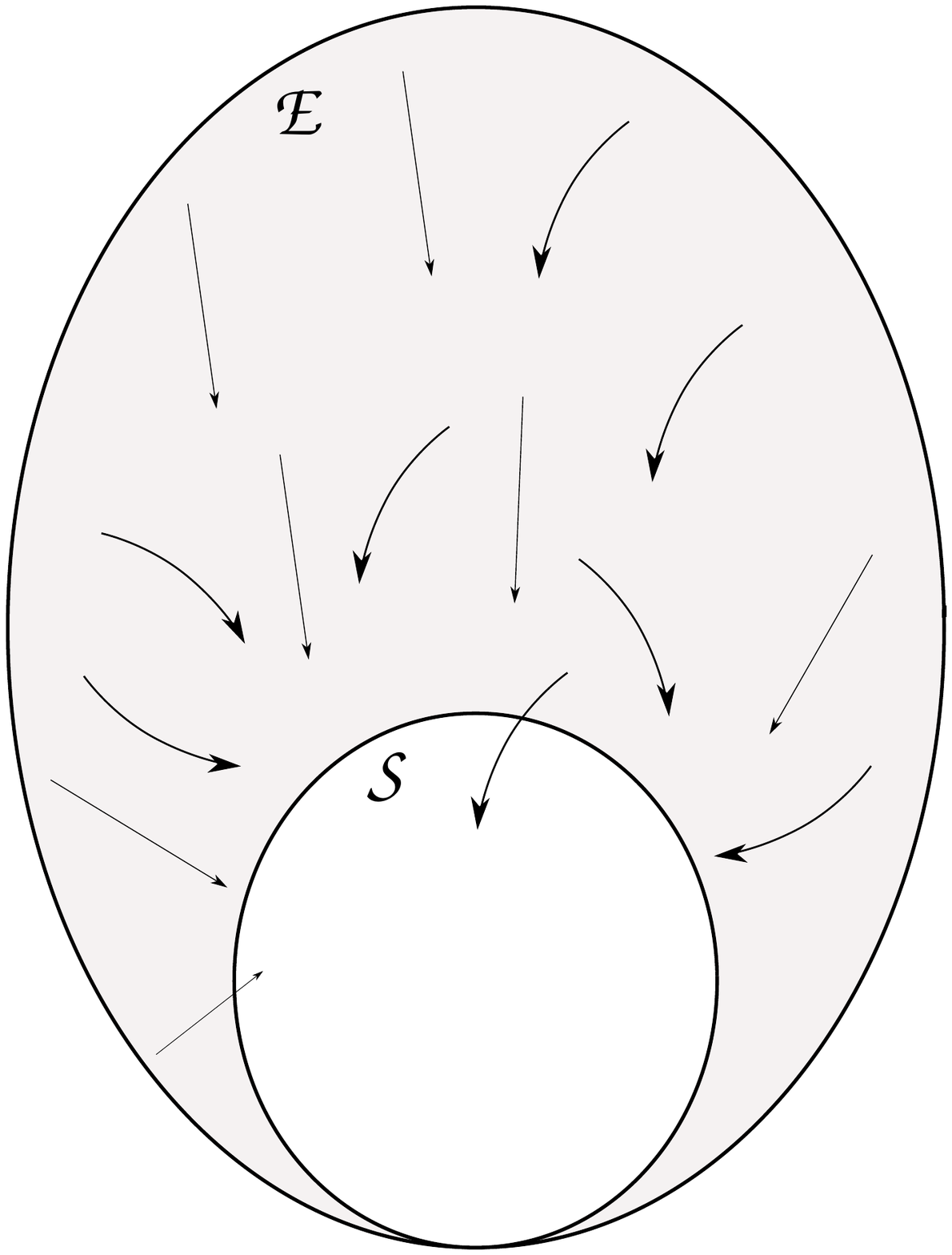}}
\qquad
\subfigure[]{\includegraphics[width=3.5cm]{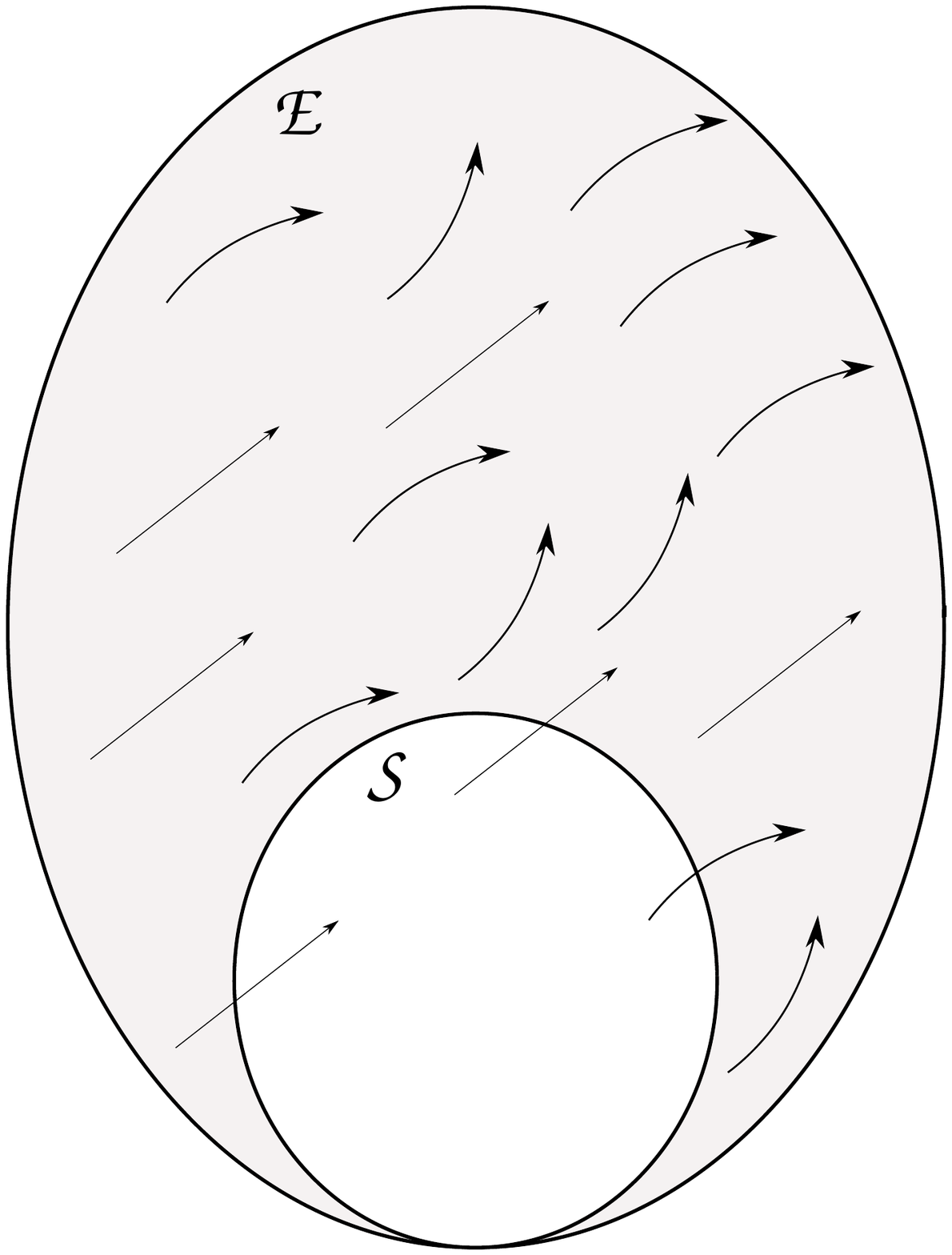}}
\caption{The arrows represent how initial states are mapped to time evolved 
ones. Figure (a) 
shows a flow directed towards a separable state, while the figure (b) shows 
the flow directed 
towards an entangled one. In fact, we should stress that it is not always true 
that 
the whole family keeps fixed some $\rho$, \emph{i.e.}, 
$\Lambda_t(\rho)=\rho, \forall t \geq 0$ for some state $\rho$.}
\label{fig2}
\end{figure}

\section*{Closed systems}

We start with the additional assumption that the bipartite system dynamics 
is induced by unitary
operations for all $t>0$ [there is some $U_{t}$ acting on $\mathcal{H}_{AB}$ 
such that
$\Lambda_{t}(\cdot)=U_{t}(\cdot)U_{t}^{*}$]. That is, the pair of systems 
may have any interaction
with each other and they can even interact with classical external sources 
(for instance, their
Hamiltonian may vary in time due to an external control of some of its 
parameters). 
Under such
conditions, FTD is a consequence of Proposition~\ref{fact} above and Theorem~\ref{thm3} (discussed 
in Section~\ref{tech}):
\begin{theorem}
 If a bipartite system have dynamics given by 
$\Lambda_{t}(\cdot)=U_{t}(\cdot)U_{t}^{*}$ for all
$t>0$, there is no FTD if, and only if, $\{\Lambda_t\}_{t \geq 0} \in 
\mathcal{F}_{\mathcal{H}_A,\mathcal{H}_B}$. 
\label{closedsystems}
\end{theorem}
\begin{proof}
Indeed, if the 
family $\Lambda_{t}$ is such
that, for some $\bar{t}>0$, $U_{\bar{t}}$ is not a local unitary operation,  
there exists an entangled state $\ket{\psi_{E}}$ 
such that
$\ket{\psi_{P}}=U_{\bar{t}}\ket{\psi_{E}}$ is a product state 
(see Corollary~\ref{corentan}). Take small enough $0<\lambda<1$ such
that $\rho_{E}=\lambda 
\frac{I}{d_{A}d_{B}}+(1-\lambda)\borb{\psi_{E}}{\psi_{E}}$ is still an
entangled state. We then have that $\Lambda(\rho_{E})=\lambda
\frac{I}{d_{A}d_{B}}+(1-\lambda)\borb{\psi_{P}}{\psi_{P}}$ is a state 
\emph{in the
interior} of the set of separable states (a convex combination of an 
arbitrary point of a convex
set with a point in the interior of it, results in an element also in its 
interior~\cite{rock}). By 
Proposition~\ref{fact}, FTD takes place.
\end{proof}

\section*{Pair of qubits}

Physically, although Theorem~\ref{closedsystems} allows for very general 
interactions between the
systems, it is restrictive with respect to their interaction with their 
environment, since this
environment must be
effectively classic. Here we greatly relax this restriction, on the expense 
of diminishing the
range of quantum systems considered.

\begin{theorem}
 If a bipartite system with Hilbert space $\mathcal{H}_{AB}$, where
\textnormal{\normalfont{dim}}$(\mathcal{H}_{A})=$\textnormal{\normalfont{dim
}}$(\mathcal{ H } _{B}
)=2$, have a dynamics
such that $\Lambda_{t}(\mathds{1})=\mathds{1}$ for all $t \geq 0$ 
(\emph{i.e.} each map is unital), there is no FTD if, and only if, 
$\{\Lambda_t\}_{t \geq 0} \in 
\mathcal{F}_{\mathcal{H}_A,\mathcal{H}_B}$. \label{opensystems}
\end{theorem}
\begin{proof}
For an 
arbitrary instant of
time $t$, we have the following four possibilities for the corresponding 
CPTP map $ \Lambda_{t}$:
$i)$ it is induced  by a local unitary operation; $ii)$ it is induced by a 
composition of a local
unitary operation with the SWAP operator; $iii)$ it is induced by a unitary 
operation which is
neither local nor the composition of a local unitary with the SWAP operator; 
$iv)$ it is not induced
by any unitary. Let us look to each situation:

$i)$ Of course, if this holds for all $t>0$, we do not have FTD. 

$iii)$ Here we can just apply 
Theorem~\ref{closedsystems} to show that there is FTD. 

$iv)$ We can find a maximally entangled state $\rho_{E}$ 
such that 
$\Lambda_{\bar{t}}(\rho_{E})$ is mixed (see Theorem~\ref{thm2}). If 
$\lambda_{-}(\rho)$ is the smallest eigenvalue of the 
partial transposition of $\rho$, we have that 
$\lambda_{-}(\rho_{E})=-\frac{1}{2}$ and 
$\lambda_{-}[\Lambda(\rho_{E})]=\delta>-\frac{1}{2}$ (see Ref.~\cite{Witnessed}). We can choose $0<p<1$ 
such that 
$\lambda_{-}[p\rho_{E}+(1-p)\frac{\mathds{1}}{4}]=p(-\frac{1}{2}-\frac{1}{4}
)+\frac{1}{4}<0$ and 
$\lambda_{-}[p\Lambda(\rho_{E})+(1-p)\frac{\mathds{1}}{4}]=p(\delta-\frac{1}
{4})+\frac{1}{4}>0$. 
That is, the initial state $p\rho_{E}+(1-p)\frac{\mathds{1}}{4}$ is 
entangled but its time evolved 
state at $\bar{t}$, $p\Lambda(\rho_{E})+(1-p)\frac{\mathds{1}}{4}$ is in the 
interior of the set of  
separable states. By Proposition~\ref{fact}, we have FTD. 

$ii)$ Finally, if this is the case, the continuity of the family of
maps allows us to conclude for the existence of a $0<\bar{t}<t$ where 
$\Lambda_{\bar{t}}$
fits in either cases $iii)$ or $iv)$, since the set of CPTP maps induced by 
such unitaries is disjoint from the set induced by local unitaries (a continuous path between
two disjoint sets must necessarily pass trough the complement of them).

\end{proof}

\section*{Example: Markovian dynamics}

A Markovian dynamics~\cite{plenioetal} is distinguished by a semi-group 
property satisfied by the family of CPTP maps:
\begin{equation}\Lambda_{t+t'}=\Lambda_{t}\circ\Lambda_{t'},\end{equation}
for all $t,t'\geq 0$. It holds then~\cite{lindblad} that the dynamics can be 
equivalently described 
by a differential equation (a Lindblad equation):
\begin{equation} 
 \frac{d\rho(t)}{dt}=-i[H,\rho]+\sum_{i=1}^{N}\left( A_{i}\rho 
A_{i}^{*}-\frac{1}{2}\{A^{*}_{i}A_{i},\rho\} \right),\label{markov}
\end{equation}
where $H$ is self-adjoint while $A_{i}$ are linear operators. 
 Lindbladian equations
can describe a plethora of physical 
phenomena, such as the dissipation of electromagnetic field modes of a cavity, spontaneous emission of atoms, 
spin dephasing due to a random magnetic field and so on. 
Therefore, despite the fact that 
the semi-group condition is somewhat restrictive, it is satisfied by many 
relevant quantum systems.  
The first term in the r.h.s. generates a unitary evolution and can usually be interpreted as the Hamiltonian evolution of the \textit{isolated} system.
The term involving the operators $A_i$ is usually called \textit{dissipator}, being responsible for the contractive part of the dynamics.

When an operator $A_{i}$ is proportional to the 
identity it does 
not contribute to the 
dynamics. Moreover, the dynamics will preserve the purity of initial states 
if, and only if, all 
operators $A_{i}$ are of such kind (that is, the dynamics is Hamiltonian):
\begin{lemma}
For $\rho(t)$, a solution of Eq.~\eqref{markov} with initial condition 
$\borb{\psi}{\psi}$, it holds 
that \text{\normalfont{lim}}$_{t\rightarrow 
0}\frac{d\text{\normalfont{Tr}}[\rho^2(t)]}{dt}=0$ for 
all $\ket{\psi}$ if, and only if, $A_{i}=\lambda_{i}I$ for 
$i=1,...,N$.\label{lemaderivada}
\end{lemma}
\begin{proof}
Indeed, for $t>0$
\begin{align*}
&\frac{d\text{\normalfont{Tr}}[\rho^2]}{dt}=2\text{\normalfont{Tr}}[\frac{
d\rho}{dt}\rho]\\
&=2\text{\normalfont{Tr}}(-i[H,\rho]\rho+\sum_{i=1}^{N}A_{i}\rho 
A_{i}^{*}\rho-\frac{1}{2}\{A^{*}_{i}A_{i},\rho\}\rho)
\end{align*}
Since $\text{\normalfont{lim}}_{t\rightarrow 0}\rho=\borb{\psi}{\psi}$, it 
follows that:
\begin{align}
\text{\normalfont{lim}}_{t\rightarrow 
0}\frac{d\text{\normalfont{Tr}}[\rho^2(t)]}{dt}=2\sum_{i=1}^{N}(|\braket{
\psi|A_{i}|\psi}|^{2}
-||A_{i}\ket { \psi} ||^{2}).\label{equa}
\end{align}
By the Cauchy-Schwarz inequality, $$|\braket{\psi|A_{i}|\psi}|^{2}\leq 
||\ket{\psi}||^{2}||A_{i}\ket{\psi}||^{2}=||A_{i}\ket{\psi}||^{2},$$
we can conclude the r.h.s of eq.~\eqref{equa} is zero iff all terms in the 
sum are 
zero and 
$\ket{\psi}\propto A_{i}\ket{\psi}$ for every $i=1,...,N$. These 
proportionality relations 
holds for all $\ket{\psi}$ if, and only if, all $A_{i}$ are proportional to 
the 
identity operator.   
\end{proof}

The above lemma shows that, for every $t>0$, the CPTP map defined 
by Eq.~\eqref{markov} is not 
induced by a unitary operation. It is also easy to check that every CPTP 
maps given by Eq.~\eqref{markov}  
is unital as long as 
$\sum_{i=1}^{N}(A_{i}A_{i}^{*}-A_{i}^{*}A_{i})=0.$ With this in hand, by 
Theorem~\ref{opensystems}, we can state:
\begin{corollary}
 If a bipartite system with Hilbert space $\mathcal{H}_{AB}$, where
\normalfont{dim}$(\mathcal{H}_{A})=\text{\normalfont{dim}}(\mathcal{H}_{B}
)=2$, have a dynamics
described by eq.~\eqref{markov}, where some 
$A_{i}$ is not a multiple of the identity and 
$\sum_{i=1}^{N}(A_{i}A_{i}^{*}-A_{i}^{*}A_{i})=0$,
there is FTD.
\end{corollary}

\section{Unital pure state preserving maps and product preserving unitaries}\label{tech}

In this section we prove some results about CPTP maps,
such as the
characterization of unital and  pure state preserving ones, which were used 
in the Section~\ref{ftd}.

Consider a bipartite quantum system with finite dimensional Hilbert space
$\mathcal{H}$. We say that a
CPTP map $\Lambda$, acting on the set of all density operators $\mathcal{D}(\mathcal{H})$, is
\emph{pure state preserving} if $\Lambda (\borb{\psi}{\psi})$ is a
pure state for
every pure state $\ket{\psi}$. 
Trivial examples of such maps are those
induced by 
unitary operations [$\Lambda (\rho)=U\rho U^{\dagger}$, for $U$ unitary acting on $\mathcal{H}$] and
the constant maps $\Lambda
(\rho)=\borb{\phi_{0}}{\phi_{0}}$ where $\ket{\phi_{0}}$ is a fixed
state. Moreover a CPTP map is said to be unital if it maps the maximally mixed
state on itself. 

\begin{theorem} 
 Every pure state preserving unital map $\Lambda:\mathcal{D}(\mathcal{H})\rightarrow
\mathcal{D}(\mathcal{H})$, where \normalfont{dim}$(\mathcal{H})=d<\infty$, is induced by a 
unitary
operation.
\end{theorem}
\begin{proof}
Take a Naimark dilation of $\Lambda$, that is, a unitary $U$ acting on a larger space 
$\mathcal{H}\otimes\mathcal{R}$ and a fixed vector $\ket{R}\in\mathcal{R}$, such that 
$\Lambda(\rho)=\text{Tr}_{\mathcal{R}}[U(\rho\otimes\borb{R}{R})U^{*}]$ for all
$\rho\in\mathcal{D}(\mathcal{H})$. 

It must be the case that $U\ket{\phi}\otimes
\ket{R}$ is a product vector for all $\ket{\phi}\in\mathcal{H}$, since otherwise
Tr$_{\mathcal{R}}[U(\borb{\phi}{\phi}\otimes \borb{R}{R})U^{*}]$ would not be a one-dimensional
projector and $\Lambda$ would not preserve pure states.

Now, if $\{\ket{\phi_{j}}\}_{j=1}^{d}$ is an orthonormal basis, we have that 
$\Lambda(\borb{\phi_{j}}{\phi_{j}})=P_{j}$ for some one-dimensional projectors  $P_{j}$.
From $\Lambda$ being 
unital, it holds that
$\Lambda(\sum_{j=1}^{d}\borb{\phi_{j}}{\phi_{j}})=\sum_{j=1}^{d}P_{j}=I$, so
the projectors $P_{j}$ must be mutually orthogonal.

With the last two paragraphs in mind it must be true that, for $j=1,...,d$, there are normalized
vectors $\ket{\psi_{j}}\in\mathcal{H}$ and $\ket{R_{j}}\in\mathcal{R}$, such that
$U\ket{\phi_{j}}\otimes\ket{R}=\ket{\psi_{j}}\otimes\ket{R_{j}}$. Moreover, the set
$\{\ket{\psi_{j}}\}_{j=1}^{d}$ must be orthonormal. On the other hand, for $j=2,...,d$,
$$U(\ket{\phi_{1}}+\ket{\phi_{j}})\otimes\ket{R}=\ket{\psi_{1}}\otimes\ket{R_{1}}+\ket{\psi_{j}}
\otimes \ket{R_{j}}.$$
For the vectors on the r.h.s of this equation being product, given that $\ket{\phi_{1}}$ is
orthogonal to $\ket{\phi_{j}}$, it must hold that $\ket{R_{j}}=z_{j}\ket{R_{1}}$ for some
$z_{j}\in\mathds{C}$ of unity modulus. If we define a unitary $V$ acting on $\mathcal{H}$ by
$V\ket{\phi_{j}}=z_{j}\ket{\psi_{j}}$ for $j=1,...,d$, we get $\Lambda(\rho)=V\rho V^{*}$ for all
density operators $\rho$.
\end{proof}

\begin{lemma}
 Let $\mathcal{H}_A,\mathcal{H}_B$ be two bi-dimensional Hilbert spaces. If 
$\ket{\phi}, \ket{\psi} \in \mathcal{H}_A \otimes 
\mathcal{H}_B$, and 
$\ket{\phi}+e^{i\theta}\ket{\psi}$ is a product vector for all $\theta \in \mathds{R}$, then 
$\ket{\phi}$ and $\ket{\psi}$ are product too. \label{emC2vale}
\end{lemma}
\begin{proof}
Let be $\ket{\psi}=a\ket{00}+b\ket{11}$ a Schmidt decomposition for $\ket{\psi}$, and 
$\ket{\phi}=\alpha\ket{00}+\beta\ket{01}+\gamma\ket{10}+\delta\ket{11}$ the expression for 
$\ket{\phi}$ with respect to the basis $\{\ket{00},\ket{01},\ket{10}, \ket{11}\}$. For arbitrary $z
\in 
\mathds{C}$, we can define the family of vectors:
\begin{subequations}
 \begin{eqnarray}
\nonumber
& \ket{z}= \ket{\phi}+z\ket{\psi} = (az+\alpha)\ket{00} + (bz+\delta)\ket{11}+ \\ \nonumber
& +\beta\ket{01}+\gamma\ket{10}.
 \end{eqnarray}
\end{subequations}
For each $z$, the above state factorizes if, 
and only if, the following determinant is zero:
$$
D=
\left| \begin{array}{cc}
(az+\alpha) & \beta  \\ 
 \gamma     & (bz+\delta) 
 \end{array} \right| = abz^2 + (a\delta+b\alpha) z + \alpha\delta + \beta\gamma.
$$
If $a,b \neq 0$ ($\emph{i.e.}$, $\ket{\psi}$ is entangled), $D$ can not be identically zero for all
values of $z$. Therefore,
 $\ket{\psi}$ must be product. By similar reasoning, we conclude $\ket{\phi}$ is also product. 
\end{proof}

\begin{lemma}
 Let $\mathcal{H}_A,\mathcal{H}_B$ be two Hilbert spaces with dimension $d 
\geq 2$. If $\ket{\phi}, \ket{\psi} \in \mathcal{H}_A \otimes 
\mathcal{H}_B$, and 
$\ket{\phi}+e^{i\theta}\ket{\psi}$ is a product state for all $\theta \in \mathds{R}$, then 
$\ket{\phi}$ and $\ket{\psi}$ are product too. \label{emCdvale}
\end{lemma}

\begin{proof}
 Let us argue by contradiction. Suppose that $\ket{\psi}$ is entangled,
 thus in the Schmidt decomposition $\ket{\psi}=\sum_{l=1}^{d}\psi_{l}\ket{ll}$ there are, at
least two indexes $l_1,l_2$ such that $\psi_{l_1},\psi_{l_2} \neq 0$. Writing 
$\ket{\phi}=\sum_{k,j}\phi_{k,j}\ket{kj}$ in the same basis as $\ket{\psi}$, 
and defining $\psi_{k,j}=\psi_k \delta_{k,j} $ we get:
\begin{subequations}
 \begin{eqnarray}
 \nonumber
& \forall \theta \in \mathds{R}: 
\ket{\theta}=\ket{\psi}+e^{i\theta}\ket{\phi}=\sum_{k,j}(\psi_{k,j}+e^{i 
\theta} \phi_{k,j})\ket{kj}.
 \end{eqnarray}
\end{subequations}
Therefore $\ket{\theta}$ is product, by hypothesis, for all $\theta \in 
\mathds{R}$. Projecting 
$\ket{\theta}$ at the subspace generated by 
$\{\ket{l_{1}l_{1}},\ket{l_{1}l_{2}},\ket{l_{2}l_{1}},\ket{l_{2}l_{2}}\}$ we obtain:
$$
 \ket{\xi_{\theta}}=\sum_{k,j \in \{l_1,l_2\}}(\psi_{k,j}+e^{i \theta} \phi_{k,j})\ket{kj}.
$$
Since $\ket{\xi_{\theta}} \in \mathds{C}^2 \otimes \mathds{C}^2$ is product for all values of 
$\theta$, we can apply Lemma~\ref{emC2vale} and obtain the desired contradiction.
\end{proof}

\begin{theorem} 
If $\Lambda$ is a unital map acting on $\mathcal{H}_{AB}=\mathds{C}^2 \otimes \mathds{C}^2$ and 
preserves the purity of maximally entangled states, then $\Lambda$ is 
induced by an 
unitary operation.\label{thm2}
\end{theorem}
\begin{proof}
 Take a representation of $\Lambda$ in terms of a unitary $U$ acting on a larger space 
$\mathcal{H}_{AB}\otimes\mathcal{H}_{R}$, such that 
$$\Lambda(\rho)=\text{Tr}_{R}[U(\rho\otimes\borb{R}{R})U^{*}],$$ where 
$\ket{R}\in\mathcal{H}_{R}$. With $U(\ket{00}\otimes \ket{R})=\ket{\psi}$ and
$U(\ket{11}\otimes \ket{R})=\ket{\phi}$, we have, for all $\theta \in \mathds{R}$:
\begin{subequations}
 \begin{eqnarray}
 \nonumber
(\ket{00}+e^{i\theta}\ket{11})\otimes \ket{R} \xmapsto{U} \ket{\psi}+e^{i \theta}\ket{\phi}.
 \end{eqnarray}
\end{subequations}
As $\Lambda$ preserves the purity of $(\ket{00}+e^{i\theta}\ket{11})$, the state 
$\ket{\psi}+e^{i \theta}\ket{\phi}$ is product for all $\theta$, with respect to 
$\mathcal{H}_{AB}\otimes\mathcal{H}_{R}$. Lemma~\ref{emCdvale} implies that $\ket{\psi}$
and $\ket{\phi}$ are both product, that is:
\begin{subequations}
 \begin{eqnarray}
  & \ket{00}\otimes\ket{R} \xmapsto{U} \ket{\psi_{00}}\otimes \ket{R_{00}} \\
  &  \ket{11}\otimes\ket{R} \xmapsto{U} \ket{\psi_{11}}\otimes \ket{R_{11}}.
 \end{eqnarray}
\end{subequations}

Let $\mathfrak{B}=\{\ket{\Psi_{\pm}},\ket{\Phi_{\pm}}\}$ be the Bell basis in 
$\mathcal{H}_{AB}$. The map $\Lambda$ satisfies:
\begin{subequations}
 \begin{eqnarray}
  \nonumber 
 & \mathbb{1}=\Lambda(\mathbb{1})=\Lambda(\borb{\Phi_{+}}{\Phi_{+}}+ \\ \nonumber
 & \borb{\Phi_{-}}{\Phi_{-}}+\borb{\Psi_ {+}}{\Psi_{+}}+\borb{\Psi_{-}}{\Psi_{-}}).
 \end{eqnarray}
\end{subequations}
Since the images $\Lambda(\borb{\Phi_{\pm}}{\Phi_{\pm}})$ and 
$\Lambda(\borb{\Psi_{\pm}}{\Psi_{\pm}})$ are $4$ unidimensional projectors ($\Lambda$ preserves 
purity of maximally entangled states)
that sum up to the identity, they must be mutually orthogonal. 

Observe that the combinations $(\ket{\psi_{00}} \otimes \ket{R_{00}}) \pm (\ket{\psi_{11}} \otimes 
\ket{R_{11}})$ must be product with respect to $\mathcal{H}_{AB}\otimes \mathcal{H}_{R}$, because 
they are images of $\ket{\Phi_{\pm}}\otimes \ket{R}$ under $U$. We state that 
$$
\ket{R_{00}}=e^{i \gamma} \ket{R_{11}}.
$$
Otherwise, $\ket{\psi_{00}} \propto \ket{\psi_{11}}$, and then 
$\Lambda(\borb{\Phi_{+}}{\Phi_{+}})=\borb{\Psi_{00}}{\Psi_{00}}=\Lambda(\borb{\Phi_{-}}{\Phi_{-}})$
contradicting the fact that $\Lambda(\borb{\Phi_{\pm}}{\Phi_{\pm}})$ are mutually orthogonal.
Again, from
\begin{subequations}
 \begin{eqnarray*}
  & \ket{01}\otimes\ket{R} \xmapsto{U} \ket{\psi_{01}}\otimes \ket{R_{01}}, \\
  &  \ket{10}\otimes\ket{R} \xmapsto{U} \ket{\psi_{10}}\otimes \ket{R_{10}},
 \end{eqnarray*}
\end{subequations}
we derive that $\ket{R_{01}}=e^{i \delta}\ket{R_{10}}$. Now, define 
$\ket{\xi}=a\ket{\Phi_{+}}+b\ket{\Phi_{-}}+c\ket{\Psi_{+}}+d\ket{\Psi_{-}}$, for a suitable 
choice of constants $a,b,c,d \neq 0$ such that $\ket{\xi}$ is maximally entangled. Therefore
\begin{subequations}
 \begin{eqnarray}
 \nonumber
 & U(\ket{\xi}\otimes \ket{R})=(a\ket{\psi_{00}}+b e^{-i\gamma}\ket{\psi_{11}}) \otimes 
\ket{R_{00}} \\ \nonumber
& + (c\ket{\psi_{01}}+d e^{-i\delta}\ket{\psi_{10}}) \otimes \ket{R_{01}},
\end{eqnarray}
\end{subequations}
and then $\ket{R_{00}}=e^{i\beta}\ket{R_{01}}$. We can define a unitary operator V, acting on 
$\mathcal{H}_{AB}$, given by:

\begin{subequations}
\begin{align}
& \ket{00} \xmapsto{V} \ket{\psi_{00}}, \\ 
& \ket{11} \xmapsto{V} e^{-i \gamma}\ket{\psi_{11}}, \\
& \ket{01} \xmapsto{V} e^{i(\delta-\beta-\gamma)}\ket{\psi_{01}}, \\ 
& \ket{10} \xmapsto{V} e^{-i(\delta+\gamma)}\ket{\psi_{10}}.
\end{align}
\end{subequations}
With this definition, we have $\Lambda(\cdot)=V(\cdot)V^{*}$.
\end{proof}

When $\mathcal{H}_{A}=\mathcal{H}_{B}$, we can define the so-called SWAP operator $S$, by
$S (\ket{\phi}\otimes \ket{\psi})=\ket{\psi}\otimes \ket{\phi}$. If the Hilbert spaces are not the
same, but have the same dimension, we can take any isomorphism $\Psi:\mathcal{H}_{A}\rightarrow
\mathcal{H}_{B}$ between them and define the operators $S_{\Psi}=(\Psi^{-1}\otimes I_{B}) \circ
S\circ (\Psi\otimes I_{B}),$ where $I_{B}$ is the identity operator on $\mathcal{H}_{B}$,
\emph{i.e}, $S_{\Psi}\ket{\phi}\otimes \ket{\psi}=\Psi^{-1}(\ket{\psi})\otimes \Psi(\ket{\phi})$
which we shall also denote by SWAP.

The theorem below characterizes unitary operations acting on composite Hilbert spaces that
preserve product vectors:

\begin{theorem} Let $U$ be a unitary operation acting on a Hilbert space
$\mathcal{H}_{A}\otimes \mathcal{H}_{B}$, where $\mathcal{H}_{A(B)}$ has finite dimension
$d_{A(B)}\geq 2$. Then $U$ is product preserving if, and only if, it is a 
local 
unitary operation or,
for the case 
\textnormal{\normalfont{dim}}$(\mathcal{H}_{A})=$\textnormal{\normalfont{dim
}} $(\mathcal { H } _ { B } )$ , a 
composition
of a local unitary
operation with a SWAP operator. \label{thm3}
\end{theorem}
\begin{proof} Consider an orthonormal basis in each space
$\{\ket{j}_{A}\}_{j=0}^{\text{dim}(\mathcal{H}_{A})-1}$,
$\{\ket{k}_{B}\}_{k=0}^{\text{dim}(\mathcal{H}_{B})-1}$. The unitary 
operation must map states 
$\ket{j}_{A}\otimes\ket{k}_{B}$ into elements $\ket{\psi_{jk}}_{A}\otimes\ket{\phi_{jk}}_{B}$,
which are mutually orthogonal. Since the images of the product vectors
$(\ket{j}_{A}+\ket{j'}_{A})\otimes\ket{k}_{B}$, that is
$\ket{\psi_{jk}}_{A}\otimes\ket{\phi_{jk}}_{B}+\ket{\psi_{j'k}}_{A}\otimes\ket{\phi_{j'k}}_{B}$
are also product vectors, we must have one of two options
\begin{subequations}
\begin{eqnarray}
&\ket{\psi_{jk}}_{A}\perp\ket{\psi_{j'k}}_{A}\text{    and  
}\ket{\phi_{jk}}_{B}\propto\ket{\phi_{j'k}}_{B}, \label{Aop1}\\
\nonumber &\text{   or}\\
&\ket{\phi_{jk}}_{B}\perp\ket{\phi_{j'k}}_{B}\text{    and
}\ket{\psi_{jk}}_{A}\propto\ket{\psi_{j'k}}_{A}. \label{Aop2}
\end{eqnarray}
\end{subequations}
For a fixed $k$, if one of the options is valid for a pair $j$ and $j'$, it must be valid for all
such pairs. Indeed, suppose that the first option is valid for, say, $j=0$
and $j'=1$ and the second for $j=0$ and $j'=2$. The
image of the product vector $(\ket{1}_{A}+\ket{2}_{A})\otimes\ket{k}_{B}$, given by
$\ket{\psi_{1k}}_{A}\otimes\ket{\phi_{1k}}_{B}+\ket{\psi_{2k}}_{A}\otimes\ket{\phi_{2k}}_{B}$
would be an entangled vector, since we would have
$\ket{\psi_{1k}}_{A}\perp\ket{\psi_{0k}}_{A},\ket{\psi_{2k}}_{A}\propto\ket{\psi_{0k}}_{A}$,
$\ket{\phi_{1k}}_{B}\propto\ket{\phi_{0k}}_{B}$ and
$\ket{\phi_{2k}}_{B}\perp\ket{\phi_{0k}}_{B}$. Therefore,
$\ket{\psi_{1k}}_{A}\perp\ket{\psi_{2k}}$ and
$\ket{\phi_{1k}}_{B}\perp\ket{\phi_{2k}}_{B}$.

$i)$ Assume that \eqref{Aop1} is true. That means that the vectors
$\ket{\phi_{jk}}_{B}$ are proportional to each other for fixed $k$, while the vectors
$\ket{\psi_{jk}}_{A}$, also for fixed $k$, form an orthonormal basis. We can write then
$U\ket{j}_{A}\otimes\ket{k}_{B}=e^{i\theta_{jk}}\ket{\psi_{jk}}_{A}\otimes\ket{\phi_{0k}}_{B}$.

If we consider the image of the vectors
$\ket{j}_{A}\otimes(\ket{k}_{B}+\ket{k'}_{B})$, we deduce that we have the following options
\begin{subequations}
\begin{eqnarray}
&\ket{\phi_{jk}}_{B}\perp\ket{\phi_{jk'}}_{B}\text{    and  
}\ket{\psi_{jk}}_{A}\propto\ket{\psi_{jk'}}_{A}, \label{Aop1a}\\
\nonumber &\text{   or}\\
&\ket{\psi_{jk}}_{A}\perp\ket{\psi_{jk'}}_{A}\text{    and
}\ket{\phi_{jk}}_{B}\propto\ket{\phi_{jk'}}_{B}. \label{Aop2a}
\end{eqnarray}
\end{subequations}
Again, similarly to what we have above, if one of the option is valid for a 
pair $k$ and $k'$, for
fixed $j$, it must be valid for all such pairs. But given that
\eqref{Aop1} is true, now only \eqref{Aop1a} can also be.
Indeed, if \eqref{Aop2a} were true, we would have, for example, the 
subspace generated
by the vectors $\{\ket{j}_{A}\otimes\ket{0}_{B}, \ket{0}_{A}\otimes\ket{k}_{B}\}$,
of dimension $\text{dim}(\mathcal{H}_{A})+\text{dim}(\mathcal{H}_{B})-1$,
mapped to the subspace $\mathcal{H}_{A}\otimes\ket{\phi_{00}}$, of dimension 
dim$(\mathcal{H}_{A})$,
contradicting the fact the $U$ is unitary.

Since we have that \eqref{Aop1a} is true, we can write
$U\ket{j}_{A}\otimes\ket{k}_{B}=e^{i\theta_{jk}}\ket{\psi_{j0}}_{A}\otimes\ket{\phi_{0k}}_{B}$.
Using this expression, and demanding that the states
$(\ket{j}_{A}+\ket{j'}_{A})\otimes(\ket{k}_{B}+\ket{k'}_{B})$ are of the product form for all pairs
$j,j'$ and $k,k'$, we obtain
$e^{i(\theta_{jk}+\theta_{j'k'})}=e^{i(\theta_{jk'}+\theta_{j'k})}$.
In particular, if $k'=j'=0$, we get $\theta_{jk}=\theta_{j0}+\theta_{0k}(\text{mod }2\pi)$, since
$\theta_{00}=0$ by construction. Finally, we have $U=U_{A}\otimes
U_{B}$ with $U_{A}\ket{j}_{A}=e^{i\theta_{j0}}\ket{\psi_{j0}}_{A}$ and
$U_{B}\ket{k}_{B}=e^{i\theta_{0k}}\ket{\phi_{0k}}_{B}$.

$ii)$ Assume that \eqref{Aop2} is true. Note firstly that it is necessary 
to have
$\text{dim}(\mathcal{H}_{A})\geq\text{dim}(\mathcal{H}_{B})$ since, for
fixed $k$, we are varying over $\text{dim}(\mathcal{H}_{A})$ orthonormal 
vectors on $A$, which
therefore give rise to a set of orthonormal vectors $\ket{\phi_{jk}}_{B}$ in $\mathcal{H}_{B}$.
So
$U(\ket{j}_{A}\otimes\ket{k}_{B})=e^{i\tilde{\theta}_{jk}}\ket{\psi_{0k}}_{A}\otimes\ket{\phi_{jk}}_
{ B
}$.
Now only the option \eqref{Aop2a} can be true, so again we have
$\text{dim}(\mathcal{H}_{B})\geq\text{dim}(\mathcal{H}_{A})$, and therefore
$\text{dim}\mathcal{H}_{A}=\text{dim}\mathcal{H}_{B}$, which allows us to 
write
$U(\ket{j}_{A}\otimes\ket{k}_{B})=e^{i\tilde{\theta}_{jk}}\ket{\psi_{0k}}_{A}\otimes\ket{\phi_{j0}}_
{ B
}$.
Considering again that the image of the states
$(\ket{j}_{A}+\ket{j'}_{A})\otimes(\ket{k}_{B}+\ket{k'}_{B})$ must be product vectors, we have
$\tilde{\theta}_{jk}=\tilde{\theta}_{j0}+\tilde{\theta}_{0k}(\text{mod
}2\pi)$. In other words $U=(U_{A}\otimes U_{B})\circ
S_{\Psi}$, where $U_{A}\ket{j}_{A}=e^{i\tilde{\theta}_{0j}}\ket{\psi_{0j}}$,
$U_{B}=e^{i\tilde{\theta}_{k0}}\ket{\phi_{k0}}_{B}$ and $\Psi \ket{k}_{A}=\ket{k}_{B}$.
\end{proof}

Putting these results together we have the following:
\begin{corollary}
 If $U$ is a unitary operator acting on a Hilbert space $\mathcal{H}_{A}\otimes \mathcal{H}_{B}$,
where
$\mathcal{H}_{A(B)}$ has finite dimension and preserves entangled states, then it is a local unitary
operation or,
for the case 
\emph{dim}$(\mathcal{H}_{A})=\text{\emph{dim}}(\mathcal{H}_{B})$, a
composition
of a local unitary
operation with a SWAP operator.\label{corentan}
\end{corollary}
\begin{proof}
If $U$ preserves entangled states, its inverse $U^{-1}$ preserves product states. From
Theorem~\ref{thm3}, there are unitaries $V_{A}$ and $V_{B}$ acting on
$\mathcal{H}_{A}$ and $\mathcal{H}_{B}$, respectively, such that $U^{-1}=V_{A}\otimes V_{B}$ or
$U^{-1}=S\circ V_{A}\otimes V_{B}$, therefore $U=U_{A}\otimes U_{B}$ or $U=U_{A}\otimes
U_{B}\circ S$, with $U_{A}=V_{A}^{-1}$ and $U_{B}=V_{B}^{-1}$. 
\end{proof}

\section{Discussion} \label{Disc}

Although we could not prove Conjecture~\ref{conjecture} in its full 
generality, we manage to do it for
some large and important families of quantum dynamics. 
They include all 
possible dynamics for a bipartite
closed system, whatever interaction the parts might have and whatever time variation their
Hamiltonian may have. 
For qubits a much larger class of dynamics possibilities were considered,
only requiring a technical condition (unitality) on CPTP maps describing 
the time evolution. 
Since the proof for qubits seems quite technical and the geometric ingredients are the same for other finite dimensions, the Conjecture that the only class of bipartite dynamics not to show FTD is the local unitaries must hold, but still demands a final proof.

The requirement of finite dimensional Hilbert spaces seems to be essential. 
Indeed, the geometrical insight is based on the fact that the set
of separable states has non-empty interior, which ceases to be true whenever one of the Hilbert
spaces is of infinite dimension~\cite{infdim}. 
Of course, even in that 
case, where generically one does not expect FTD, 
many physically relevant dynamics actually can show it, such as those 
preserving Gaussian states~\cite{terra3}.  

Other situation where topology changes, and consequently entanglement dynamics changes, is when one restricts to pure states.
There, the set of separable states (indeed, product states) has empty interior.
For these systems, FTD can only happen if ``hand tailored'', {\textit{e.g.}}: starting from an entangled state, some family of global unitaries is applied up to a time when the state is product, from this time on, only local unitaries are applied.
This is clearly not generic in the set of dynamics. 

As a last commentary, it is natural to remember that for practical implementations of quantum information processing, it is important to fight against FTD.
Our results about the genericity of FTD do not make this fight impossible.
Even for dynamics where FTD does happen, is is natural to search for 
initial 
states where can be avoided, or, at least, delayed~\cite{review,primeiramorte,ugo1,ugo2}.

\begin{acknowledgments}
We thank Marcelo F. Santos for interesting discussions and CNPq, FAPEMIG and CAPES for financial support. This work is part
of the Brazilian National Institute of Quantum Information.
\end{acknowledgments}

\end{document}